\documentclass[a4paper]{article}
\usepackage[english]{babel}
\usepackage{amsmath}
\usepackage{amsthm}
\usepackage{amssymb}
\usepackage{geometry} 

\theoremstyle{plain}
    \newtheorem{theorem}{Theorem}
    \newtheorem{consequence}{Corollary}
\theoremstyle{remark}
    \newtheorem{remark}{Remark}
  \theoremstyle{definition}
    \newtheorem{definition}{Definition}
    
\newcommand{\R}{\mathbb{R}}
\newcommand{\T}{\mathrm{T}}
\newcommand{\so}{\mathfrak{so}}
\newcommand{\SO}{\mathrm{SO}}
\newcommand{\Orth}{\mathrm{O}}
\newcommand{\U}{\mathrm{U}}
\newcommand{\J}{\mathcal{J}}

\title{A note on relative equilibria of multidimensional rigid body}
\author{Anton Izosimov}
\begin{document}
\maketitle
\abstract{It is well known that a rotation of a free generic three-dimensional rigid body is stationary if and only if it is a rotation around one of three principal axes of inertia.
As it was noted by many authors, the analogous result is true for a multidimensional body: a rotation is stationary if and only if it is a rotation in the principal axes of inertia, provided that the eigenvalues of the angular velocity matrix are pairwise distinct. 
However, if some eigenvalues of the angular velocity matrix of a stationary rotation coincide, then it is possible that this rotation has a different nature. A description of such rotations is given in the present paper.}

\section{Introduction}
Speaking informally, a free multidimensional rigid body is simply a rigid body rotating in multidimensional space without action of any external forces (i.e., by inertia).\par
Let us first discuss a three-dimensional free rigid body (the so-called Euler case in the rigid body dynamics). A good model for such a body is a book or a parallelepiped shaped box.\par
Throw the book in the air spinning it in arbitrary direction. If we neglect the gravity force, then what we get is exactly the Euler case. \par
Note that  a general trajectory of a body is not a rotation in the usual sense. At each moment of time our body is indeed rotating around some axis, but this axis is changing as time goes. What we are interested in, are the relative equilibria of the system, i.e. such trajectories for which the axis of rotation remains fixed. Such rotations are also called stationary.\par
It is well known that a generic three-dimensional rigid body (i.e. a body with pairwise distinct principal moments of inertia) admits three and only three stationary rotations: these are the rotations around three principal axes of inertia. If we deal with a parallelepiped shaped body, then these axes coincide with the axes of symmetry.\par
What we want to do, is to generalize this result to the case of a multidimensional body. \par
The equations of a free multidimensional rigid body were first written by F.Frahm \cite{Frahm}. V.Arnold \cite{Arnold} wrote these equations in the form of Euler equations on $\so(n)^{*}$ and generalized them to the case of an arbitrary Lie algebra. A possibility to generalize Euler equations  to the multidimensional case was also mentioned by H.Weyl \cite{Weyl}:
\begin{quote}``The above treatment of the problem of rotation may, in contradistinction to the usual method, be transposed, word for word, from three-dimensional space to multidimensional spaces. This is, indeed, irrelevant in practice. On the other hand, the fact that we have freed ourselves from the limitation to a definite dimensional number and that we have formulated physical laws in such a way that the dimensional number appears \textbf{accidental} in them, gives us an assurance that we have succeeded fully in grasping them mathematically.''\end{quote}.\par 
The equations of a free multidimensional rigid body are famous for being a completely integrable system. As it was shown by S.Manakov \cite{Manakov}, the system admits an $L-A$ pair with a spectral parameter. This allowed him to write down integrals and to show that the system is integrable in Riemann $\theta$-functions. Complete integrability in the Liouville sense was proved by A.Fomenko and A.Mishchenko \cite{MF2, MF} and by T.Ratiu \cite{Manakov2}.
\section{Rotation of an $n$-dimensional body}
First we shall discuss how an $n$-dimensional body may rotate. At each moment of time, $\R^n$ is decomposed into the sum of $m$ pairwise orthogonal two-dimensional planes $\Pi_1, \dots, \Pi_m$ and an $n-2m$-dimensional space $\Pi_0$ orthogonal to all these planes:
$$
\R^n = \left(\bigoplus_{i=1}^{m} \Pi_i \right)\oplus \Pi_0.
$$
 There is an independent rotation in each of the planes $\Pi_1, \dots, \Pi_m$, while $\Pi_0$ is fixed. This is just a reformulation of the theorem about canonical form of a skew-symmetric operator. Note that $\Pi_0$ may be zero in the even-dimensional case, which means that there are no fixed axes.\par
A rotation is stationary if all the planes $\Pi_0, \dots, \Pi_m$ don't change with time (this condition automatically implies that the velocities of rotations are also constant).\par
A natural generalization of the three-dimensional result would be the following: A rotation of a generic multidimensional rigid body is stationary if and only if the planes $\Pi_0, \dots, \Pi_m$ are spanned by the principal axes of inertia.
 This is true provided that the angular velocities of rotations in planes  $\Pi_1, \dots, \Pi_m$ are pairwise distinct (see \cite{Marshall, Spiegler, Ratiu}), however not true in general. The existence of relative equilibria for which $\Pi_0, \dots, \Pi_m$ are not spanned by the principal axes of inertia was probably first noted in \cite{Ratiu}. In this note we give a complete description of such equilibria. Although this result is very simple, we could't find it in the literature.
 
 \section{The equations}
 The motion of a free multidimensional rigid body is described by the Euler-Arnold (or Euler-Frahm) equations on $\so(n)^{*}$ (identified with $\so(n)$). These equations have the form
 \begin{align}\label{EulerEq}
 \begin{cases}
 \dot M &= [M, \Omega]\\
 M &= \Omega J + J \Omega,
 \end{cases}
 \end{align}
 where 
 \begin{itemize}
\item  $M \in \so(n)^{*}$ is a skew-symmetric matrix, called the angular momentum matrix;
\item  $J$ is a symmetric matrix (see Remark \ref{JRemark});
\item $\Omega$ is a skew-symmetric matrix, called the angular velocity matrix. It is uniquely defined by the relation
$$
M = \Omega J + J \Omega.
$$
 \end{itemize}
 \begin{remark}
Since the map $\J \colon \so(n) \to \so(n)$ given by the formula
 $$
\J(\Omega) = \Omega J + J \Omega
 $$ 
 is invertible, our equations can be rewritten in the $\Omega$-coordinates:
 $$
 	\dot \Omega = \J^{-1}([\J(\Omega), \Omega]).
 $$
 However, the explicit formula for $\J^{-1}$ is complicated, therefore it is convenient to introduce the variable $M$ and write down the equations in the form (\ref{EulerEq}).
 \end{remark}
 \begin{remark}\label{JRemark}
 	In the multidimensional case $J$ is sometimes referred to as the ``inertia tensor'', which seems to be not very precise, because in the three-dimensional case the inertia tensor is not $J$, but the map
	$\J\colon \so(3) \to \so(3)$ given by the formula
 $$
\J(\Omega) = \Omega J + J \Omega.
 $$ 
 These two tensors (in the three-dimensional case) have common eigenvectors, but different eigenvalues: the eigenvalues of $\J$ are pairwise sums of the eigenvalues of $J$. 
  \end{remark}
 \begin{remark}
Note that equations (\ref{EulerEq}) describe only the dynamics of the angular velocity matrix. If we want to recover the dynamics in the whole phase space $\T^*\SO(n)$, we should add Poisson equations
$$
	\dot X = X\Omega, \mbox{ where } X \in \SO(n).
$$
However, we will only be interested in reduced dynamics, given by (\ref{EulerEq}). Note that relative equilibria of a rigid body is nothing else but the equilibrium points of (\ref{EulerEq}).
\end{remark}

\section{Description of relative equilibria}\label{dREq}
 \begin{theorem}
Consider the system of Euler-Arnold equations
  \begin{align*}
 \begin{cases}
 \dot M &= [M, \Omega]\\
 M &= \Omega J + J \Omega.
 \end{cases}
 \end{align*}
 Suppose that $J$ has pairwise distinct eigenvalues. Then $M$ is an equilibrium point of the system if and only if there exists an orthonormal basis such that $J$ is diagonal, and $\Omega$ is block-diagonal of the following form
\begin{align}\label{omegaForm}
\Omega =  \left(\begin{array}{cccccc}\omega_1A_1 &  & & & & \\ & \ddots & & & & \\ &  & \omega_{k}A_k & & & \\  & & & 0 & & \\ & & & & \ddots & \\ & & & & & 0 \end{array}\right),
\end{align}
 where $A_{i} \in \so(2m_i) \cap \SO(2m_i)$ for some $m_i > 0$, and $\omega_i$'s are distinct positive real numbers.\par
 Form (\ref{omegaForm}) is unique up to a permutation of blocks.
 \end{theorem}
 \begin{proof}
  	We have $
	[M, \Omega] = [ \Omega J + J \Omega, \Omega] = [J, \Omega^{2}]
	$,
	therefore $M$ is an equilibrium if and only if $\Omega^{2}$ commutes with $J$.\par
	Assume we have a basis such that $J$ is diagonal and $\Omega$ has the form (\ref{omegaForm}).	 Then 
	$$
		A_i^2 = -A_iA_i^{\mathrm{t}} = -E
	$$
	and
		 \begin{align}\label{OmegaSq}
	\Omega^{2} =  \left(\begin{array}{cccccc}-\omega_1^2E &  & & & & \\ & \ddots & & & & \\ &  & -\omega_{k}^2E & & & \\  & & & 0 & & \\ & & & & \ddots & \\ & & & & & 0 \end{array}\right).
	\end{align}
	Therefore, $[\Omega^{2}, J] = 0$, and our point is an equilibrium point.\par
	Vice versa, let  $[\Omega^{2}, J] = 0$. We shall prove that there exists an orthonormal basis such that $J$ is diagonal and $\Omega$ has the form (\ref{omegaForm}).\par
	First find an orthonormal basis such that $J$ is diagonal. $\Omega^{2}$ is diagonal in this basis as well, since $J$ has pairwise distinct eigenvalues. Also note that the diagonal entries of $\Omega^2$ in this basis are non-positive, because $\Omega$ is skew-symmetric and has only pure imaginary or zero eigenvalues.	 Now, by a permutation of basis vectors, we can bring $\Omega^{2}$ to the form (\ref{OmegaSq}) where all $\omega_{i}$'s are positive and pairwise distinct.\par
	Since $\Omega^2$ is in the form (\ref{OmegaSq}) and $[\Omega^2,\Omega] = 0$, $\Omega$ has the form
	\begin{align*}
		\Omega = \left(\begin{array}{ccc}B_{1} &  &  \\ & \ddots &  \\ &  & B_{k+1}\end{array}\right),
	\end{align*}	  
	 where $B_{i} ^{2} = -\omega_{i}^{2}E$ for $i \leq k$, and $B_{k+1}^2 = 0$. \par
	 Since $B_{k+1}^2 = 0$ and $B_{k+1}$ is skew-symmetric, we have $B_{k+1} = 0$.
	 For $i \leq k$ set
	 $$
	 A_{i} = \frac{1}{\omega_{i}}B_{i}.
	 $$ 
	 On the one hand, $A_{i} \in \so(l_i)$ for some $l_i$. On the other hand
	 $$
	 A_{i}A_{i}^{\mathrm{t}} = -A_{i}^{2} =  -\frac{1}{\omega_{i}^{2}}B_{i}^{2} = E,
	 $$
	 which means that $A_{i} \in \SO(l_i)$. But $\so(l_i) \cap \SO(l_i)$ is empty for odd $l_i$, therefore $l_i = 2m_i$ and $\Omega$ has the form (\ref{omegaForm}), q.e.d.
 \end{proof}
 \begin{remark}
 	Note that $ \so(2m) \cap \SO(2m)$ is the homogeneous space $\Orth(2m) / \U(m)$ which is identified with the space of  complex structures compatible with the standard euclidian metrics.
 \end{remark}
 \begin{consequence}
 	A relative equilibrium of Euler-Arnold equations is defined by:
	\begin{enumerate}\item Choosing a decomposition
	$$
\R^n = \left(\bigoplus_{i=1}^{k} \Pi_i \right)\oplus \Pi_0,
$$
where all $\Pi_i$'s are spanned by the main axes of inertia and all $\Pi_i$'s for $i>0$ are even-dimensional.
\item Assigning an angular velocity $\omega_i > 0$ and a complex structure compatible with the Euclidian metrics to each $\Pi_i$ for $i > 0$.
\end{enumerate}
 \end{consequence}
 \begin{consequence}[Well-known, see \cite{Marshall, Spiegler, Ratiu}]\label{regEquilibria}
 	Suppose that $M$ is a relative equilibrium, and all eigenvalues of $J$ are pairwise distinct. Moreover, let all non-zero eigenvalues of $\Omega$ be pairwise distinct. Then there exists an orthonormal basis such that $J$ is diagonal, while $\Omega$ and $M$ are block-diagonal with two-by-two blocks on the diagonal. \par
	In other words, a stationary rotation with pairwise distinct eigenfrequencies is a rotation in the principal axes of inertia.
 \end{consequence}
 \begin{remark}
 Sometimes the result of Corollary \ref{regEquilibria} is formulated in the following way: if a relative equilibrium belongs to a regular (co)adjoint orbit, then it is a rotation in the principal axes of inertia. This is not very precise, because ``belongs to a regular (co)adjoint orbit'' means that $M$ is regular. However regularity of $M$ doesn't imply regularity of $\Omega$, and vice versa.
    \end{remark}
    \newpage
   We suggest the following
 \begin{definition}
 	We will say that an equilibrium $M$ is \textit{regular} if there exists an orthonormal basis such that $J$ is diagonal and $\Omega$ is block-diagonal with two-by-two blocks on the diagonal (i.e. this equilibrium is a rotation in the principal axes of inertia). Otherwise, we will say that $M$ is \textit{exotic}.
 \end{definition}
Corollary \ref{regEquilibria} says that all stationary rotations with pairwise distinct eigenfrequencies are regular.\par
In \cite{biham} A.Bolsinov and A.Oshemkov study those equilibria of system (\ref{EulerEq}) which are equilibrium points simultaneously for all the integrals of the system. It is proven that the set of such equilibria coincides with the set of regular equilibria in our terminology. Consequently, for each exotic equilibrium $M$ we can find an integral $f$ such that the hamiltonian vector field $v$ generated by $f$ doesn't vanish at $M$. All the points belonging to the trajectory of $v$ passing through $M$ will be equilibrium points of Ê(\ref{EulerEq}).
Consequently, exotic equilibria are not isolated on the coadjoint orbits of $\so(n)$, but form smooth submanifolds of equilibrium points (while regular equilibria are, on the contrary, always isolated on a given orbit). This can be used to prove that exotic equilibria are always Lyapunov unstable.
\begin{remark}
The problem of stability of relative equilibria of a multidimensional rigid body was studied by many authors. See \cite{Marshall, Ratiu, Casu2, Casu, Spiegler}.
\end{remark}
\bibliographystyle{unsrt}  
\bibliography{Diss} 

\end{document}